\documentclass[12pt]{article}
\newcommand{\calc}{\ensuremath{{\cal C}}}

\newcommand{\sat}{{\rm SAT}}
\newcommand{\boldsat}{{\bf SAT}}
\newcommand{\sharpsat}{{\rm \#SAT}}
\newcommand{\sharpp}{{\rm \#P}}
\newcommand{\littlep}{{p}}
\newcommand{\manyone}{\ensuremath{\,\leq_{m}^{{\littlep}}\,}}
\newcommand{\p}{{\rm P}}
\newcommand{\np}{{\rm NP}}
\newcommand{\conp}{{\rm coNP}}
\newcommand{\sigmastar}{\ensuremath{\Sigma^\ast}}
\usepackage{url}
\usepackage{booktabs}
\usepackage{amsmath}
\usepackage{amssymb}
\usepackage{amsthm}
\usepackage{tikzit}

\newcommand{\true}{\ensuremath{\mathrm{True}}}
\newcommand{\false}{\ensuremath{\mathrm{False}}}

\newtheorem{theorem}{Theorem}%

\newtheorem{lemma}{Lemma}
\newtheorem{fact}{Fact}
\newtheorem{challengetheorem}{Challenge Problem}
\newtheorem{definition}{Definition}
\newtheorem{example}{Example}

\begin{document}
\sloppy
\title{The Power of Self-Reducibility:\\Selectivity,
  Information, and Approximation\thanks{This \mbox{arXiv.org}
    report is a preliminary
    version of a chapter that will
    appear in the in-preparation book
    \emph{Complexity and Approximation},
    eds.~Ding-Zhu Du and Jie Wang, Springer.  This article was written
    in part while on sabbatical at Heinrich Heine University D\"usseldorf,
    supported in part by a Renewed Research Stay grant from the
  Alexander von Humboldt Foundation.}}
\date{February 21, 2019; revised March 14, 2019}
\author{Lane A. Hemaspaandra\\Department of Computer Science\\University of Rochester\\Rochester, NY 14627 USA} %
\maketitle

\emph{In memory of Ker-I Ko, whose indelible
contributions to computational
complexity included important work (e.g.,~\cite{ko-moo:j:approx,ko:j:maximum-value-CAREFUL-selman-did-left-cuts-first-see-comment,ko:j:self-reducibility-CAREFUL-selman-did-left-cuts-first-see-comment,ko:j:helping})
on each of this chapter's topics: self-reducibility, selectivity,
information, and approximation.}

\begin{abstract}
This chapter provides a hands-on tutorial on the important
technique known as self-reducibility.  Through
a series of ``Challenge Problems''
that are theorems that the reader will---after being given definitions and
tools---%
try to prove, the tutorial will ask the reader not to read proofs that
use self-reducibility, but rather to \emph{discover} proofs that use
self-reducibility.  In particular, the chapter will seek to guide the
reader to the discovery of proofs of four interesting
theorems---whose
focus areas range from selectivity to information to
approximation---from the literature, whose proofs draw on
self-reducibility.

The chapter's goal is to allow interested readers to add
self-reducibility to their collection of proof tools.  The chapter
simultaneously has a related but different goal, namely, to provide a
``lesson plan'' (and a coordinated set of slides is available online to
support this use~\cite{hem:url:power-of-self-reducibility-slides})
for a lecture to a two-lecture series
that can be given to undergraduate students---even those with no
background other than basic discrete mathematics and an 
understanding of what
polynomial-time computation is---to immerse them in hands-on proving,
and by doing that, to serve as an invitation to them to take
courses on Models of Computation or Complexity Theory.
\end{abstract}

\section{Introduction}
Section~\ref{ss:groups}
explains the two quite different audiences
that this chapter is intended for, and for each describes
how that group might use the chapter.  If you're not
a computer science professor it would
make sense to skip Section~\ref{sss:profs}, and if you
are a computer science professor you might at least on a first
reading choose to skip Section~\ref{sss:students}

Section~\ref{ss:self-and-sat} introduces the type of self-reducibility
that this chapter will focus on, and the chapter's central
set,
$\sat$
(the satisfiability problem for propositional Boolean formulas).

\subsection{A Note on the Two Audiences, and How to Read This Chapter}\label{ss:groups}
This chapter is unusual in that it has two intended audiences, and
those audiences differ dramatically in their amounts of background in
theoretical computer science.

\subsubsection{For Those Unfamiliar with Complexity Theory}\label{sss:students}
The main intended audience is those---most especially
young
students---who are not yet familiar with complexity theory, or perhaps
have not even yet taken a models of computation course.  If that
describes you, then this chapter is intended to be a few-hour tutorial
immersion in---and invitation to---the world of theoretical computer science
research.  As you go through this tutorial,
you'll try to solve---hands-on---research issues that are sufficiently
important that their original solutions
appeared in
some of theoretical computer science's best conferences and journals.

You'll be given the definitions and some tools before being asked to
try your hand at finding a solution to a problem.  And with luck, for
at least a few of our four
challenge
problems, you will find a
solution to the problem.  Even if you don't find a solution for
the given problem---and the problems increase in difficulty and
especially the later ones require
bold, flexible exploration to find
possible paths to the solution---the fact that you have spent time trying
to solve the problem  will give you more insight into the solution when
the solution is
then presented in this chapter.

A big-picture goal here is to make it clear that doing theoretical
computer science research is often about playful, creative, flexible
puzzle-solving.  The underlying hope here is that many people who
thought that theoretical computer science was intimidating and
something that they could never do or even understand will realize that
they can do theoretical computer science and perhaps even that they
(gasp!)~\emph{enjoy} doing theoretical computer science.  

The four problems also are tacitly bringing out a different issue, one
more specifically about complexity.  Most people, and even most
computer science professors, think that complexity theory is
extraordinarily abstract and hard to grasp.  Yet in each of our four
challenge problems, we'll see that doing complexity is often extremely
concrete, and in fact is about building a program that solves a given
problem.  Building programs is something that many people already have done, e.g.,
anyone who has taken an introduction to programming course or
a data structures course.  The only
difference in the programs one builds when doing proofs in complexity
theory is that the programs one builds typically draw on some
hypothesis that provides a piece of the program's action.  For
example, our fourth challenge problem will be to show that if a
certain problem is easy to approximate, then it can be solved exactly.
So your task, when solving it, will be to write a program that
exactly solves the problem.  But in writing your program you will
assume that you have as a black box that you can draw on as a program (a
subroutine) that given an instance of the problem gives an approximate
solution.

This view that complexity is largely about something that is quite
concrete, namely building programs, in fact is the basis of an entire
graduate-level complexity-theory textbook~\cite{hem-ogi:b:companion},
in which the situation is described as follows:

\begin{quote}
Most people view complexity theory as an
arcane
realm populated by pointy-hatted (if not indeed pointy-headed)
sorcerers stirring cauldrons of recursion theory with wands of
combinatorics, while chanting incantations involving complexity
classes whose very names contain hundreds of characters and sear the
tongues of mere mortals.  This stereotype has sprung up in part due to
the small amount of esoteric research that fits this bill, but the
stereotype is more strongly attributable to the failure of complexity
theorists to communicate in expository forums the central role that
algorithms play in complexity theory.
\end{quote}

\paragraph{Expected Background} To keep this chapter as accessible as
possible, the amount of expected background has been kept quite small.
But
there are
some types of background that are being assumed here.  The reader is
assumed to know following material, which would typically be
learned within about the first
two courses of most computer science departments' introductory course
sequences.
\begin{enumerate}
\item  What a polynomial is.

  As an example,
$p(n) = n^{1492} + 42n^{42} + 13$ is a polynomial; $e(n) = 2^n$ is
not.

\item
What it means for a set or function to be computable in
polynomial time, i.e., to be computed in time polynomial in the number
of bits in the input to the problem.  The class of all sets that can
be computed in polynomial time is denoted P, and is one of the most
important classes in computer science.

As an example, the set of all
positive integers that are multiples of 10 is a set that belongs to
$\p$.

\item Some basics of logic such as the meaning
  of quantifiers ($\exists$ and $\forall$) and what a (propositional) Boolean
  formula is.

  As an example of the latter,
the formula $x_1 \land (x_2 \lor \overline{x_3})$ is a such a formula,
and evaluates as $\true$---with each of $x_1$, $x_2$, and $x_3$ being variables
whose potential values are $\true$ or $\false$---exactly if $x_1$ is
$\true$ and either $x_2$ is $\true$ or the negation of $x_3$ is $\true$.
\end{enumerate}

If you have that background in hand, wonderful!  You have the
background to tackle this chapter's puzzles and challenges, and please
(unless you're a professor thinking of modeling a lecture series
on this chapter) skip from here right on to Section~\ref{ss:self-and-sat}.

\subsubsection{For Computer Science Professors}\label{sss:profs}
Precisely because this chapter is designed to be accessible and fun
for students who don't have a background in theoretical computer science,
the chapter avoids---until Section~\ref{c:add-complexity}---trying to abstract
away from the focus on $\sat$.  In particular, this chapter either avoids mentioning
complexity class names such as NP, coNP, and PSPACE, or at least, when it
does mention them, uses phrasings such as ``classes known as'' to make clear
that the reader is not expected to possess that knowledge.

Despite that, computer science professors are very much an intended
audience for this chapter, though in a way that is reflecting the fact
that the real target audience
for these challenges 
is young students.
In particular, in addition to providing a tutorial introduction for
students, of the flavor described in Section~\ref{sss:students}, this
chapter also has as its goal to provide to you, as a teacher, a
``lesson plan'' to help you offer in your course a one- or two-day
lecture (but really hands-on workshop) sequence\footnote{To cover all
  four problems would take two class sessions.  Covering just
  the first two or
  perhaps three of the problems could be done in a single 75-minute
  class session.} in which you present the definitions and tools of
the first of these problems, and then ask the students to break into
groups and in groups spend about 10--25 minutes working on solving the
problem,\footnote{In this chapter, since student readers of the
  chapter will be working as individuals, I suggest to the reader, for
  most of the problems, longer amounts of time.  But in a classroom
  setting where students are working in groups, 10--25 minutes may be
  an appropriate amount of time; perhaps 10 minutes for the first
  challenge problem, 15 for the
  second, 15 for the third, and 25 for the fourth.  You'll need to
  yourself judge the time amounts that are best, based on your knowledge of your students.  For
  many classes, the just-mentioned times will not be enough.
  Myself, I try to keep
  track of whether the groups seem to have found an answer, and I will
  sometimes stretch out the time window if many groups seem to be
  still working intensely and with interest.  Also, if TAs happen to
  be available who don't already know the answers, I may assign them to
  groups so that the class's groups will have more experienced members, though
the TAs do know to guide rather than dominate a group's discussions.}  and then you ask
whether some group has found a solution and would like to present it
to the class, and if so you and the class listen to and if needed
correct the solution (and if no group found a solution, you and the class
will work together to reach a solution).
And then you go on to similarly treat the other
three problems, again with the class working in teams.  This provides
students with a hands-on immersion in team-based, on-the-spot
theorem-proving---something that most students never get in class.  I've done this
in classes of size up to 79 students, and they love it.  The approach
does not treat them as receptors of information lectured at them, but
rather shows them that they too can make discoveries---even ones that
when first obtained appeared in such top forums as \emph{CCC} (the
yearly \emph{Computational Complexity Conference}), \emph{ICALP} (the
yearly \emph{International Colloquium on Automata, Languages, and
  Programming}), the journal \emph{Information and Computation}, and
\emph{SIAM Journal on Computing}.

To support this use of the chapter as a teaching tool in class, I have
made publicly available a set of \LaTeX/Beamer slides that can be used
for a one- or two-class hands-on workshop series on this chapter.  The
slides are available
online~\cite{hem:url:power-of-self-reducibility-slides}, both as
pdf slides and, for teachers who might wish to modify the slides,
as 
a zip archive of
the source files.

Since the slides might be used in courses
where students already do know of such classes as NP and coNP, the
slides don't defer the use of those classes as
aggressively 
as this
chapter itself does.  But the slides are designed so that the mention
of the connections to those classes is parenthetical (literally
so---typically a parenthetical, at the end of a theorem statement,
noting the more general application of the claim to all of NP or
all of $\conp$), and those
parentheticals can be simply skipped over.  Also, the slides define on the
fly both NP and coNP, so that if you do wish to cover the more general
versions of the theorems, the slides will support that too.

The slides don't themselves present the solutions to
Challenge
Problems 1, 2, or~3.  Rather, they assume that one of the class's
groups will present a solution on the board (or document camera) or
will speak the solution with the professor acting as a scribe at the
board or the document camera.
Challenge
Problems 1, 2, and~3
are doable enough
that usually at least one group will either have solved the question,
or at least will made enough progress
that,
with
some help from classmates or some hints/help from the
professor, a solution can quickly be reached building on
the group's work.  (The professor
ideally should
have read the solutions in this chapter to those problems, so that
even if a solution isn't reached or almost reached by the students on one
or two of those problems, the professor can provide a solution at the board or
document camera.  However, in the ideal case, the solutions of those
problems will be heavily student-driven and won't need much, if any,
professorial steering.)

Challenge Problem~4 is sufficiently hard that the slides do include
both a slide explaining why a certain very natural approach---which is
the one students often (quite reasonably) try to make work---cannot
possibly work, and thus why the approach that the slides gave to
students as a gentle, oblique
hint may be the way to go, and then the slides present a
solution along those lines.

The difficulty of Challenge Problem~4 has
a point.  Although
this chapter is trying to show students that they
\emph{can} do theory research, there is also an obligation not to give
an artificial sense that all problems are easily solved.  Challenge
Problem~4 shows students that some problems can have multiple twists
and turns on the way to finding a solution.  Ideally, the students
won't be put off by this, but rather will appreciate both that solving
problems is something they can do, and that in doing so one may well
run into obstacles that will take some out-of-the-box thinking to try
to get around---obstacles that might take not minutes of thought but rather
hours or days or more, as well as working closely with others
to share ideas as to what might work.

\subsection{Self-Reducibility and SAT}\label{ss:self-and-sat}
Now that you have read whatever part of
Section~\ref{ss:groups} applied to you, to get the lay of the land
as to what this chapter is trying to provide you, 
let us discuss the
approach that will be our lodestar throughout this
chapter.

One of the most central paradigms of computer science is ``divide and
conquer.''  Some of the most powerful realizations of that approach
occur though the use of what is known as self-reducibility, which is
our chapter's central focus.

Loosely put, a set is self-reducible if any membership question
regarding the set can be easily resolved by asking (perhaps more than
one) membership questions about smaller strings.

That certainly
divides, but does it conquer?

The answer varies greatly depending on the setting.  Self-reducibility
itself, depending on which polynomial-time variant one is looking at,
gives upper bounds on a set's complexity.  However, those
bounds---which in some cases are the complexity classes known as NP
and PSPACE---are nowhere near to putting the set into deterministic
polynomial time (aka,~$\p$).

The magic of self-reducibility, however, comes when one adds another
ingredient to one's stew.
Often, one can prove that if a set 
is self-reducible \emph{and has some other property regarding its
  structure}, then the set \emph{is} feasible, i.e., is in
deterministic polynomial time.

This tutorial will ask the reader to---and help the reader
to---discover for him- or herself the famous proofs of three such
magic cases (due to Selman, Berman, and Fortune), and then of a fourth
case that is about the ``counting'' analogue of what was described in the
previous paragraph.

Beyond that, I hope you'll keep the tool/technique
of self-reducibility in mind
for the rest of your year, decade, and lifetime---and on each new
challenge will spend at least a
few moments asking, ``Can self-reducibility
play a helpful role in my study of this problem?''  And with luck, sooner
or later, the answer may be, ``Yes!
Goodness\dots~what a surprise!''

Throughout this chapter, our model language (set) will be ``$\sat$,'' i.e., the
question of whether a given Boolean formula, for some way of assigning
each of its variables to $\true$ or to $\false$, evaluates to $\true$.
$\sat$ is a central problem in computer science, and possesses a
strong form of self-reducibility.  As a quiet bonus, though we won't
focus on this in our main traversal of the problems and their
solutions, $\sat$ has certain ``completeness'' properties that make
results proven about $\sat$ often yield results for an entire
important class of problems known as the ``NP-complete'' sets; for
those interested in that, Section~\ref{c:add-complexity}, ``Going Big:
 Complexity-Class Implications,''
briefly covers that broader view.

\section{Definitions Used Throughout: SAT and Self-Reducibility}

The game plan of this chapter, as mentioned above, is this:
For each of the four challenge problems (theorems),
you will be given definitions and some other background or tools.
Then the challenge problem (theorem) will be stated, and you'll be
asked to try to solve it, that is, you'll be asked to prove the theorem.
After you do, or after you hit a wall so completely that you feel you
can't solve the theorem even with additional
time, you'll read a proof of
the theorem.  Each of the four challenge problems has an appendix
presenting a proof of the result.

But before we start on the problems, we need to define
$\sat$ and discuss its self-reducibility.

\begin{definition}
  $\sat$ is the set of all satisfiable (propositional) Boolean formulas.
\end{definition}

\begin{example}
  \begin{enumerate}
  \item
    $x \land \overline{x} \not\in \sat$, since neither of the two possible
  assignments to $x$ causes the formula to evaluate to $\true$.
\item  $(x_1
 \land x_2 \land \overline{x_3}) \lor (x_4 \land \overline{x_4}) \in \sat$,
since
that formula evaluates to
$\true$ under at least one of the eight possible ways that the four variables
can each be assigned to be $\true$ or $\false$.  For example, when we take $x_1 = x_2 = x_4 = \true$ and $x_3 = \false$, the formula evaluates to $\true$.
\end{enumerate} 
\end{example}
$\sat$ has the following ``divide and conquer'' property.

\begin{fact}[2-disjunctive length-decreasing self-reducibility]
  Let $k \geq 1$.  Let $F(x_1, x_2, \dots,x_k)$ be a
 Boolean formula (without loss of generality, assume that each of the variables
  actually
  occurs in the formula). Then
  \[
    F(x_1, x_2, \dots,x_k) \in \sat \Longleftrightarrow
    \bigl(F(\true, x_2, \dots,x_k) \in \sat
    \,\lor\,
    F(\false, x_2, \dots,x_k) \in \sat\bigr).\]
\end{fact}

The above fact says that
\emph{SAT is self-reducible} (in particular, in the
lingo, it says that $\sat$ is
2-disjunctive
length-decreasing self-reducible).

Note: We won't at all focus in this chapter
on details of the encoding of
formulas and other objects.  That indeed is an issue if one wants to
do an utterly detailed discussion/proof.  But the issue is not a
particularly interesting one, and certainly can be put to the side
during a first traversal, such as that which this chapter is inviting
you to make.

\paragraph{A Bit of History}
In this chapter, we typically won't focus much on references.  It is best
to immerse oneself in the challenges, without getting
overwhelmed with a lot of detailed historical context.  However, so that those
who are interested in history can have some sense of the history, and so that
those who invented the concepts and proved the theorems are property
credited, we will for most sections have an ``A Bit of History''
paragraph that extremely briefly gives literature citations and
sometimes a bit of history and context.
As to the present section,
self-reducibility
dates back to
the 1970s, and in particular is due to the work of
Schnorr~\cite{sch:c:self-reducible}
and Meyer and Paterson~\cite{mey-pat:t:int}.

\section{Challenge Problem 1: Is SAT even \protect\textit{Semi}-Feasible?%
}
Pretty much no one believes that
$\sat$ has a polynomial-time decision algorithm, i.e.,
that $\sat\in\p$~\cite{gas:j:third-p-vs-np-poll}.  This section
asks you to show that it even is unlikely that $\sat$ has a
polynomial-time \emph{semi-decision} algorithm---a
polynomial-time algorithm that, given any two
formulas, always outputs one of them and does so in such a way
that if at least one of the input formulas is satisfiable then
the formula that is output is satisfiable.

\subsection{Needed Definitions}

  A set $L$ is said to be feasible (in the sense of belonging to P) if
  there is a polynomial-time algorithm that decides membership in $L$.

    A set is said to be semi-feasible (aka P-selective) if
    there is a polynomial-time algorithm that semi-decides membership,
    i.e., that given any two strings, outputs one that is
    ``more likely'' to be in the set
    (to be formally cleaner, since the probabilities
    are all 0 and 1 and can tie, what is really meant
    is ``no less likely'' to be in the set).  The following
    definition makes this formal.  (Here and elsewhere,
    $\Sigma$ will denote our (finite) alphabet and $\sigmastar$ will denote the
    set of finite strings over that alphabet.)

    \begin{definition}\label{d:p-sel}
      A set $L$ is $\p$-selective if there exists a polynomial-time function,
      $f: \sigmastar \times \sigmastar \rightarrow \sigmastar$ such
      that, 
      \[ (\forall a,b \in \sigmastar)[
      f(a,b) \in \{a,b\} \land
      \bigl(  \{a,b\} \cap L \neq \emptyset \implies f(a,b) \in L\bigr)].\]
    \end{definition}

    It is known that some P-selective sets can be very hard.  Some
    even have the property known as being ``undecidable.''  Despite
    that, our first challenge problem is to prove that $\sat$ cannot
    be P-selective unless $\sat$ is outright easy computationally,
    i.e., $\sat\in\p$.  Since it is close to an article of faith in
    computer science that $\sat\not\in\p$, showing that some
    hypothesis implies that $\sat \in \p$ is considered, with the full
    weight of modern computer science's current understanding and
    intuition, to be extremely strong evidence that the hypothesis is
    unlikely to be true.  (In the lingo, the hypothesis is implying
    that $\p = \np$.  Although it is possible that $\p=\np$ is true,
    basically no one believes that it
    is~\cite{gas:j:third-p-vs-np-poll}.  However, the issue is the
    most important open issue in applied mathematics, and there is
    currently a \$1,000,000 prize for whoever resolves the
    issue~\cite{cla:url:p-vs-np-prize}.)

\paragraph{A Bit of History}
Inspired by an analogue in computability theory,
P-selectivity was defined by Alan L. Selman in a seminal series of
papers~\cite{sel:j:pselective-tally,sel:j:some-observations-psel,sel:j:reductions-pselective,sel:j:ana}, which
included a proof of our first challenge
theorem.  The fact, alluded to above, that P-selective sets can be
very hard is due to Alan L. Selman's above work and to the work of the
researcher in whose memory this chapter is written, Ker-I
Ko~\cite{ko:j:maximum-value-CAREFUL-selman-did-left-cuts-first-see-comment}.
In that same paper, Ko also did very important early work
showing that P-selective sets are unlikely to have what is known as
``small circuits.''  For those particularly interested
in the P-selective sets, they and their nondeterministic
cousins are the subject of a book,
\emph{Theory of Semi-Feasible
Algorithms}~\cite{hem-tor:b:semifeasible-computation}.

\subsection{ Can SAT Be P-Selective?}

\begin{challengetheorem}
  (Prove that) if $\sat$ is $\p$-selective, then $\sat \in \p$.
\end{challengetheorem}

Keep in mind that what you should be trying to do is this.
You may assume that $\sat$ is P-selective.  So you may act as if you have
in hand a polynomial-time computable function, $f$, that in the sense of
Definition~\ref{d:p-sel} shows that $\sat$ is P-selective.  And your
task is to give a polynomial-time algorithm for $\sat$, i.e.,
a program that in time polynomial in the number of bits in its input
determines whether the input string belongs to $\sat$.  (Your algorithm
surely will be making calls to $f$---possibly quite a few calls.)

So that you have them easily at hand while working on this,
here are some of the key definitions and tools that you might want
to draw on while trying to prove this theorem.

\begin{description}
  \item[SAT] $\sat$ is the set of all satisfiable (propositional) Boolean formulas.

    \item[Self-reducibility]   Let $k \geq 1$.  Let $F(x_1, x_2, \dots,x_k)$ be a
  Boolean formula (without loss of generality assume that each of the variables
  occurs in the formula). Then $
    F(x_1, x_2, \dots,x_k) \in \sat \Longleftrightarrow
    \bigl(F(\true, x_2, \dots,x_k) \in \sat
    \lor
    F(\false, x_2, \dots,x_k) \in \sat\bigr)$.

  \item[P-selectivity]       A set $L$ is $\p$-selective if there exists a polynomial-time function,
      $f: \sigmastar \times \sigmastar \rightarrow \sigmastar$ such
      that, 
      $ (\forall a,b \in \sigmastar)[
      f(a,b) \in \{a,b\} \land
      \bigl(  \{a,b\} \cap L \neq \emptyset \implies f(a,b) \in L\bigr)]$.
    \end{description}

    My suggestion to you would be to work on proving this theorem
    until either you find a proof, or you've put in at least 20 minutes
    of thought, are stuck, and don't think that more time will be helpful.

    When you've reached one or the other of those states, please go on to
    Appendix~\ref{s:solution-1} to read a proof of this theorem.
    Having that proof will help you check whether your proof (if you
    found one) is
    correct,
    and if you did not find a proof,
    will show you a proof.  Knowing the answer to this challenge
    problem before going on to the other three challenge problems is
    important, since an aspect of this problem's solution will
    show up in the solutions to the other challenge problems.

    I've put the solutions in separate appendix sections
so that you can avoid
    accidentally seeing them before you wish to.  But please do
    (unless you are completely certain that your solution to the first
    problem is correct) read the solution for the first problem before
    moving on to the second problem.  If you did not find a proof for
    this first challenge problem, don't feel bad; everyone has days
    when we see things and days when we don't.  On the
    other hand, if you did find a
    proof of this first challenge theorem, wonderful, and if you felt
    that it was easy, well, the four problems get steadily harder, until by
    the fourth problem almost anyone would have to work very, very
    hard and be a bit lucky to find a solution.

  \section{Challenge Problem 2: Low Information Content and SAT, Part 1:
Can SAT Reduce to a Tally Set?}

Can $\sat$
have low information
  content?  To answer that, one needs to formalize what notion of low
  information content one wishes to study.  There are many such notions,
  but a particularly important one is 
  whether a given set can ``many-one polynomial-time reduce'' to
  a tally set (a set over a 1-letter alphabet).

\paragraph{A Bit of History}
Our second challenge theorem was stated and proved by
Piotr Berman~\cite{ber:c:relate}.  Berman's paper started a remarkably
long and productive line of work, which we will discuss in more
detail
in the ``A Bit of History'' note accompanying the third
challenge problem.  That same note will 
provide pointers to surveys of that line of work,
for 
those interested in additional reading.

\subsection{Needed Definitions}

$\epsilon$ will denote the empty string.

\begin{definition}
  A set $T$ is a tally set if $T \subseteq \{\epsilon,0,00,000,\dots\}$.
\end{definition}

\begin{definition}\label{d:reduce}
  We say that $A \manyone B$ ($A$ many-one polynomial-time reduces to $B$)
  if there is a polynomial-time computable 
function $g$ such that
  \[(\forall x \in \sigmastar)[x\in A \iff g(x) \in B].\]
\end{definition}
Informally, this says that $B$ is so powerful that each membership
query to $A$ can be efficiently transformed into a membership
query to $B$ that gets the same answer as would the question regarding
membership in~$A$.

\subsection{Can SAT Reduce to a Tally Set?}

\begin{challengetheorem}
  (Prove that) if there exists a tally set $T$
  such that $\sat \manyone T$, then $\sat \in \p$. 
\end{challengetheorem}

Keep in mind that what you should be trying to do is this.
You may assume that there exists a tally set $T$ such that 
$\sat \manyone T$.  You may not assume that you have a polynomial-time
algorithm for $T$; you are assuming that $T$ exists, but for all we know,
$T$ might well be very hard.
On the other hand, you \emph{may}
assume that you have in hand a polynomial-time computable function $g$ that reduces
from $\sat$ to $T$ in the sense of Definition~\ref{d:reduce}.
(After all, that reduction is (if it exists) a finite-sized program.)
Your
task here is to give a polynomial-time algorithm for $\sat$, i.e.,
a program that in time polynomial in the number of bits in its input
determines whether the input string belongs to $\sat$.  (Your algorithm
surely will be making calls to $g$---possibly quite a few calls.)

So that you have them easily at hand while working on this,
here are some of the key definitions and tools that you might want
to draw on while trying to prove this theorem.

\begin{description}
  \item[SAT] $\sat$ is the set of all satisfiable (propositional) Boolean formulas.

    \item[Self-reducibility]   Let $k \geq 1$.  Let $F(x_1, x_2, \dots,x_k)$ be a
  Boolean formula (without loss of generality assume that each of the $x_i$ actually
  occurs in the formula). Then $
    F(x_1, x_2, \dots,x_k) \in \sat \Longleftrightarrow
    \bigl(F(\true, x_2, \dots,x_k) \in \sat
    \lor
    F(\false, x_2, \dots,x_k) \in \sat\bigr)$.

\item[Tally sets]
A set $T$ is a tally set if $T \subseteq \{\epsilon,0,00,000,\dots\}$.

\item[Many-one reductions]
We say that $A \manyone B$ 
if there is a polynomial-time computable function $g$ such that
$(\forall x \in \sigmastar)[x\in A \iff g(x) \in B]$.
\end{description}

    My suggestion to you would be to work on proving this theorem
    until either you find a proof, or you've put in at least 30 minutes
    of thought, are stuck, and don't think that more time will be helpful.

    When you've reached one or the other of those states, please go on to
    Appendix~\ref{s:solution-2} to read a proof of this theorem.
    The solution to the third challenge problem is an extension of
    this problem's solution, so 
    knowing the answer to this challenge
    problem before going on to the third  challenge problem is
    important.

  \section{Challenge Problem 3: Low Information Content and SAT, Part 2: Can $\overline{\boldsat}$ Reduce to a Sparse Set?}

  This problem challenges you to show that even a class of sets that
  is far broader than the tally sets, namely, the so-called sparse
  sets, cannot be reduced to from $\overline{\sat}$ unless $\sat\in\p$.

\paragraph{A Bit of History}
This third challenge problem was stated and proved by Steve
Fortune~\cite{for:j:sparse}.  It was another step in what was a long
line of advances---employing more and more creative and sometimes
difficult proofs---that eventually led to the understanding that, unless
$\sat \in \p$, no sparse set can be hard for $\sat$ even with respect
to extremely flexible types of reductions. 
The most famous result within this line is known as
Mahaney's Theorem: If there is a sparse set $S$  such
that $\sat \manyone S$, then $\sat \in \p$~\cite{mah:j:sparse-complete}.
There are many surveys of
the just-mentioned line of work,
e.g.,~\cite{mah:b:sparse,mah:b:icss,you:j:sparse,hem-ogi-wat:c:sparse}.
The currently strongest result in that line is due
to Gla{\ss}er~\cite{gla:t:sparse} (see the survey/treatment of 
that in~\cite{gla-hem:j:clarityII}, and see also the 
results of Arvind et
al.~\cite{arv-han-hem-koe-loz-mun-ogi-sch-sil-thi:b:sparse}).

\subsection{Needed Definitions}

    Let $\|S\|$ denote the cardinality of set $S$, e.g., $\|\{\epsilon,0,0,0,00\}\| = 3$.

    For any set $L$, let $\overline{L}$ denote the complement of $L$.

 Let $|x|$ denote the length string $x$, e.g., $|\text{moon}| = 4$.

 \begin{definition}\label{d:sparse}
   A set $S$ is sparse  if there exists a polynomial $q$
   such that, for each natural number $n\in \{0,1,2,\dots\},
  $ it holds that
  \[\| \{ x \mid x \in S  \land |x| \leq n\} \| \leq q(n).\]
\end{definition}
Informally put,
  the sparse sets are the sets whose number of strings up to a given
  length is at most polynomial.
  $\{0,1\}^*$ is, for example, not a sparse
  set, since up to length $n$ it has $2^{n+1}-1$ strings.
  But all tally sets are sparse, indeed all via the bounding polynomial
  $q(n) = n+1$.

\subsection{Can $\overline{\boldsat}$ Reduce to a Sparse  Set?}

\begin{challengetheorem}
  (Prove that) if there exists a sparse set $S$
  such that $\overline{\sat} \manyone S$, then $\sat \in \p$.
\end{challengetheorem}

Keep in mind that what you should be trying to do is this.  You may
assume that there exists a sparse set $S$ such that $\overline{\sat} \manyone S$.
You may not assume that you have a polynomial-time algorithm for $S$; you
are assuming that $S$ exists, but for all we know, $S$ might well be very
hard.
On the other
hand, you \emph{may} assume that you have in hand a polynomial-time
computable function $g$ that reduces from $\overline{\sat}$ to $S$ in the sense
of Definition~\ref{d:reduce}.  (After all, that reduction is---if it exists---a
finite-sized program.)
And you may assume that you have in hand a polynomial that upper-bounds
the sparseness of $S$ in the sense of Definition~\ref{d:sparse}.
(After all, one of the countably infinite list of simple polynomials
$n^k +k$---for $k = 1,2,3,\dots$---will provide such an upper bound,
if any such polynomial upper bound exists.)
Your task here is to give a polynomial-time
algorithm for $\sat$, i.e., a program that in time polynomial in the
number of bits in its input determines whether the input string belongs to
$\sat$.  (Your algorithm surely will be making calls to $g$---possibly
quite a few calls.)

One might wonder why I just said that you should
build a polynomial-time algorithm for $\sat$,
given that the theorem speaks of $\overline{\sat}$.
However, since it is clear that
$\sat \in \p \iff \overline{\sat} \in \p$ (namely, given a
polynomial-time algorithm for $\sat$,
if we
simply
reverse
the answer on each input, then we now have a 
polynomial-time algorithm for $\overline{\sat}$), it is legal to focus
on $\sat$---and most people find doing so more natural and intuitive.

Do be careful here.  Solving
this challenge problem
may take an ``aha!\ldots~insight''~moment.
Knowing the solution to Challenge Problem 2 will be a
help here, but even with that knowledge in hand
one hits an obstacle.  And then the challenge
is to find a way around that obstacle.

So that you have them easily at hand while working on this,
here are some of the key definitions and tools that you might want
to draw on while trying to prove this theorem.
\begin{description}
\item[SAT and \boldmath$\rm\overline{SAT}$] $\sat$ is the set of all
  satisfiable (propositional) Boolean formulas.
  $\overline{\sat}$
  denotes the complement of $\sat$.%

 \item[Self-reducibility]   Let $k \geq 1$.  Let $F(x_1, x_2, \dots,x_k)$ be a
  Boolean formula (without loss of generality assume that each of the $x_i$ actually
  occurs in the formula). Then $
    F(x_1, x_2, \dots,x_k) \in \sat \Longleftrightarrow
    \bigl(F(\true, x_2, \dots,x_k) \in \sat
    \lor
    F(\false, x_2, \dots,x_k) \in \sat\bigr)$.

\item[Sparse sets]
  A set $S$ is sparse if there exists a polynomial $q$
  such that, for each natural number $n\in \{0,1,2,\dots\}$, it holds that
  $\| \{ x \mid x \in S  \land |x| \leq n\} \| \leq q(n)$.

\item[Many-one reductions]
We say that $A \manyone B$ 
if there is a polynomial-time computable function $g$ such that
$(\forall x \in \sigmastar)[x\in A \iff g(x) \in B]$.
\end{description}

    My suggestion to you would be to work on proving this theorem
    until either you find a proof, or you've put in at least 40 minutes
    of thought, are stuck, and don't think that more time will be helpful.

    When you've reached one or the other of those states, please go on to
    Appendix~\ref{s:solution-3} to read a proof of this theorem.

    \section{\boldmath Challenge Problem 4: Is \#SAT as Hard to (Enumeratively) Approximate as
      It Is to Solve Exactly?}
  This final challenge is harder than the three previous  ones. 
  To solve it, you'll have to
  have multiple insights---as to what
  approach to use, what building blocks to use, and how to use them.

  The problem is sufficiently hard that the solution is structured to
  give you, if you did not solve the problem already, a second bite at
  the apple!  That is, the solution---after discussing why the problem
  can be hard to solve---gives a very big hint, and then invites you to re-try
  to problem with that hint in hand.

\paragraph{A Bit of History}
The function
$\sharpsat$, the counting version of $\sat$, will play a central role
in this  challenge problem.  $\sharpsat$
was introduced and
studied by Valiant~\cite{val:j:permanent,val:j:enumeration}.
This final 
challenge problem, its proof (including the lemma given
in the solution I give here and the proof of that lemma),
and the notion of enumerators
and enumerative approximation are due to Cai and
Hemachandra~\cite{cai-hem:j:enum}.  The challenge problem is a weaker
version of the main result of that paper, which proves the result for
not just 2-enumerators but even for sublinear-enumerators; later work 
showed that the result even holds for all polynomial-time
computable enumerators~\cite{%
cai-hem:j:approx2}.

\subsection{Needed Definitions}

$|F|$ will denote the length of (the encoding of) formula $F$.

    $\sharpsat$ is the function that given as input a
    Boolean formula $F(x_1,x_2,\dots,x_k)$---without loss of generality assume that each
    of the variables occurs in $F$---outputs the number of
    satisfying assignments the formula has (i.e., of the $2^k$
    possible assignments of the variables to $\true$/$\false$, the number of
    those under which $F$ evaluates to $\true$; so the output will be
    a natural number in the interval $[0,2^k]$).  For example,
    $\sharpsat(x_1 \lor x_2) = 3$ and 
    $\sharpsat(x_1 \land \overline{x_1}) = 0$.

    \begin{definition} We say that $\sharpsat$ {has a polynomial-time
        $2$-enumerator} (aka, is polynomial-time 2-enumerably
      approximable) if there is a polynomial-time computable function $h$ such
      that on each input $x$,
      \begin{enumerate} \item $h(x)$ outputs a list of two
        (perhaps identical) natural numbers, and
        \item $\sharpsat(x)$
          appears in the list output by $h(x)$.
        \end{enumerate}
      \end{definition}
So a 2-enumerator $h$ outputs a 
    list of (at most) two
    candidate values for the value of $\sharpsat$ on the given
    input, and
    the actual output is always somewhere in that list.  This notion
    generalizes in the natural way to other list cardinalities, e.g.,
        $1492$-enumerators and, for each $k\in \{1,2,3,\dots\}$,
    $\max(1, |F|^k)$-enumerators.

\subsection{Food for Thought}\label{ss:food}

 You'll certainly want to use some analogue of the key
      self-reducibility observation, except now respun by you to
      be about the number of solutions of a formula and how it
      relates to or is determined by the number of solutions of
      its two ``child'' formulas.

But doing that is just the first step your quest.
      So\ldots~please play around with
      ideas and approaches.  Don't be afraid to be bold and ambitious.
      For example, you  might say ``Hmmmm, if we could do/build
      XYZ (where perhaps XYZ might be some particular insight about
      combining formulas), that would be a powerful tool in solving
      this, and I suspect we can do/build XYZ\@.''  And then you
      might want to work both on building XYZ
      and on showing in detail how, if you did
      have tool XYZ in hand, you could use it to show the theorem.

\subsection{Is \#SAT as Hard to (Enumeratively) Approximate as It Is to Solve Exactly?}

\begin{challengetheorem}[Cai and Hemachandra]
  (Prove that) 
  if $\sharpsat$ has a polynomial-time $2$-enumerator, then there is
  a polynomial-time algorithm for $\sharpsat$.
\end{challengetheorem}

Keep in mind that what you should be trying to do is this.  You may
assume that you have in hand a polynomial-time 2-enumerator for
$\sharpsat$.
Your task here is to give a polynomial-time
algorithm for $\sharpsat$, i.e., a program that in time polynomial in the
number of bits in its input determines the number of satisfying assignments
of 
the (formula encoded by the) input string.
(Your algorithm surely will be making calls to the 2-enumerator---possibly
quite a few calls.)

Do be careful here.  Proving this may take about three 
``aha!\ldots~insight''~moments;
Section~\ref{ss:food} gave
slight hints regarding two of those.

So that you have them easily at hand while working on this,
here are some of the key definitions and tools that you might want
to draw on while trying to prove this theorem.

\begin{description}
  \item[\#SAT]
    $\sharpsat$ is the function that given as input a
    Boolean formula $F(x_1,x_2,\dots,x_k)$---without loss of generality assume that each
    of the variables occurs in $F$---outputs the number of
    satisfying assignments the formula has (i.e., of the $2^k$
    possible assignments of the variables to $\true$/$\false$, the number of
    those under which $F$ evaluates to $\true$; so the output will be
    a natural number in the interval $[0,2^k]$).
    For example,
    $\sharpsat(x_1 \lor x_2) = 3$ and 
    $\sharpsat(x_1 \land \overline{x_1}) = 0$.

  \item[Enumerative approximation]
   We say
    that $\sharpsat$ has a polynomial-time $2$-enumerator (aka, is polynomial-time
    2-enumerably approximable) if there
    is a polynomial-time computable function $h$ such that on each input $x$,
    (a)~$h(x)$ outputs a list of two (perhaps identical) natural numbers, and
    (b)~$\sharpsat(x)$ appears in the list output by $h(x)$.

\end{description}

    My suggestion to you would be to work on proving this theorem
    until either you find a proof, or you've put in at least 30--60 minutes
    of thought, are stuck, and don't think that more time will be helpful.

    When you've reached one or the other of those states, please go on to
    Appendix~\ref{s:solution-4}, where you will find first 
    a discussion of what the most tempting dead end here is, why it
    is a dead end, and a tool that will help you avoid the dead end.
    And then you'll be urged to attack the problem again with that extra
    tool in hand.

\section{Going Big: Complexity-Class Implications}\label{c:add-complexity}

During all four of our challenge problems, we focused just on the
concrete problem, $\sat$, in its language version or in its counting
analogue, $\sharpsat$.

However,
the challenge results in fact apply to broader classes of
problems.  Although we (mostly) won't prove those broader results in this
chapter, this section will briefly note some of those (and the reader
may well be able to in most cases easily fill in the proofs).  
The original papers, cited in the ``A Bit of History'' notes,
are an excellent source to go to for more coverage.
None of the claims below, of course, are due to the present
tutorial paper, but rather they are generally right from the
original papers.
Also often of use
for a gentler treatment than the original papers is the textbook,
\emph{The Complexity Theory Companion}~\cite{hem-ogi:b:companion}, in which
coverage related to our four problems can be found in, respectively,
chapters 1, 1~[sic], 3, and (using a different technique and
focusing on a concrete but different target problem) 6.

Let us define the complexity class NP by
$\np = \{ L \mid L \manyone \sat\}$.  NP more commonly is defined as
the class of sets accepted by nondeterministic polynomial-time Turing
machines; but that definition in fact yields the same class of sets as
the alternate definition just given, and would require a detailed
discussion of what Turing machines are.

Recall that $\overline{L}$ denotes the complement of $L$.
Let us define the complexity
class $\conp$ by $\conp = \{ L \mid \overline{L} \in \np\}$.

A set $H$ is said to be hard for a class $\cal C$
if for each set $L \in \cal C$ it holds that $L \manyone H$.  If in
addition $H \in \calc$, then we say that $H$ is $\calc$-complete.  It
is well known---although it takes quite a bit of work to show and
showing this was one of the most important steps in the
history of computer science---that $\sat$ is NP-complete~\cite{coo:c:theorem-proving,lev:j:universal,kar:b:reducibilities}.

The following theorem follows easily from our first challenge theorem,
basically because if some NP-hard set is P-selective, that causes
$\sat$ to be P-selective.  (Why?  Our P-selector function for $\sat$
will simply polynomial-time reduce each of its two inputs to the NP-hard
set, will
run that set's P-selector function on those two strings, and then
will select as the more likely to belong to $\sat$ whichever input string
corresponded to the selected string, and for definiteness will choose
its first argument in the degenerate case where both its arguments map
to the same string.)

\begin{theorem}
  If there exists an $\np$-hard, $\p$-selective set, then $\p = \np$.
\end{theorem}
The converse of the above theorem also holds, since if
$\p=\np$ then $\sat$ and indeed all of NP is P-selective, since
P sets unconditionally are P-selective.

The following theorem follows easily from our second challenge theorem.
\begin{theorem}
  If there exists an $\np$-hard tally set, then $\p = \np$.
\end{theorem}
The converse of the above theorem also holds.

The following theorem follows easily from our third challenge theorem.
\begin{theorem}
  If there exists a $\conp$-hard sparse set, then $\p = \np$.
\end{theorem}
The converse of the above theorem also holds.

To state the complexity-class analogue of the fourth challenge problem
takes a bit more background, since the result is about
function classes rather than language classes.

There is a complexity class,
which we will not define here,
defined by Valiant
and known as $\sharpp$~\cite{val:j:permanent}, that is the set of
functions that count the numbers of accepting paths of what are known
as nondeterministic polynomial-time Turing machines.

Metric reductions give a reduction notion 
that applies to the case of functions rather
than languages, and are defined as follows.
A function
$f: \Sigma^* \rightarrow
\{0,1,2,\dots\}
$ is said to polynomial-time metric
reduce to a function $g: \Sigma^* \rightarrow
\{0,1,2,\dots\}
$ if there
exist two polynomial-time computable functions, $\varphi$ and $\psi$,
such that
$(\forall x \in \Sigma^*)[f(x) =
\psi(x,g(\varphi(x)))]$~\cite{kre:j:optimization}.
(We are assuming
that our output natural numbers are naturally coded in binary.)
We say a function $f$ is hard for
$\sharpp$
with respect to
polynomial-time metric reductions
if for every
$f' \in
\sharpp$ it holds that $f'$ polynomial-time metric reduces
to $f$; if
in addition $f \in 
\sharpp$, we say that $f$ is $\sharpp$-complete with
respect to polynomial-time metric reductions.

With that groundwork in hand, we can now state the analogue,
for counting classes, of our fourth challenge theorem.
Since we have
not defined $\sharpp$ here, we'll state the theorem both in terms of
$\sharpsat$ and in terms of $\sharpp$
(the two statements below in fact
turn out to be
equivalent).

\begin{theorem}
\begin{enumerate}
\item   If there exists a function $f$
  such that $\sharpsat$ polynomial-time metric reduces to $f$
and $f$ has a $2$-enumerator, 
then there is
a polynomial-time algorithm for $\sharpsat$.
\item
  If there exists a function that is $\sharpp$-hard with respect
  to polynomial-time metric reductions and has a $2$-enumerator, 
then there is
  a polynomial-time algorithm for $\sharpsat$.
\end{enumerate}
\end{theorem}
The converse of each of the theorem parts also holds.  The above
theorem parts (and their converses) even hold if one asks not about
2-enumerators but rather about polynomial-time enumerators that have
no limit on the number of elements in their output lists (aside from the polynomial
limit that is implicit from the fact that the enumerators have only
polynomial
time to write their lists).

\section{Conclusions}
In conclusion, 
self-reducibility
provides a powerful tool with applications across a broad range of settings.

Myself, I have found self-reducibility and its
generalizations to be useful in
understanding topics ranging from election
manipulation~\cite{hem-hem-men:c:search-versus-decision}
to backbones of and backdoors to Boolean
formulas~\cite{hem-nar:c:backbones-opacity,hem-nar:c:backdoors-opacity}
to the complexity of sparse sets~\cite{hem-sil:j:easily-checked},
space-efficient language recognition~\cite{hem-ogi-tod:j:sc},
logspace computation~\cite{hem-jia:j:logspace},
and approximation~\cite{hem-zim:j:balanced,hem-hem:j:holes-and-immunity}.

My guess and hope is that perhaps you too may find self-reducibility
useful in your future work.
    That is,  please, if it is not already there, consider adding this tool
    to \emph{your} personal research toolkit: When you face a problem,
    think (if only for a moment) whether the problem happens to be one where
    the concept of self-reducibility will help you gain insight.
    Who knows?  One of these years, you might be
    happily surprised in finding that your answer to such a question
    is~``Yes!''

    \section*{Acknowledgments}
    I am grateful to the students and faculty at the computer science
    departments of RWTH Aachen University,
Heinrich Heine University D\"usseldorf,
and the University of Rochester.  I ``test drove'' this chapter at
each of those schools in the form of a lecture or lecture series.
Particular thanks go to
Peter
Rossmanith, J\"org Rothe,
and
Muthu Venkitasubramaniam,
who invited me to speak, and to 
Gerhard Woeginger
regarding the counterexample in Appendix~\ref{s:solution-4}.
My warm thanks 
to Ding-Zhu Du, Bin Liu, and Jie Wang for inviting me to contribute to
this project that they have organized in memory of the wonderful Ker-I
Ko, whose work contributed so richly
    to the beautiful, ever-growing tapestry that is complexity theory.

    \bigskip

    \bigskip

  \appendix
\noindent{\Large\bf Appendices}\nopagebreak

\nopagebreak\section{Solution to Challenge Problem
  1}\label{s:solution-1} Before we start on the proof, let us put up a
figure that shows the flavor of a structure that we will use to help
us understand and
exploit $\sat$'s self-reducibility.  The structure is known as the
self-reducibility tree of a formula.  At the root of
this tree sits the
formula.
At the next level as the root's children, we have the formula with its first
variable assigned to $\true$ and to $\false$.  At the  level below that,
we have the two formulas from the second level, except with each of
\emph{their} first variables (i.e., the second variable of the
original formula) assigned to both $\true$ and $\false$.  
Figure~\ref{f:self-red} shows the self-reducibility tree of a
two-variable formula.
    \begin{figure}[tbp]
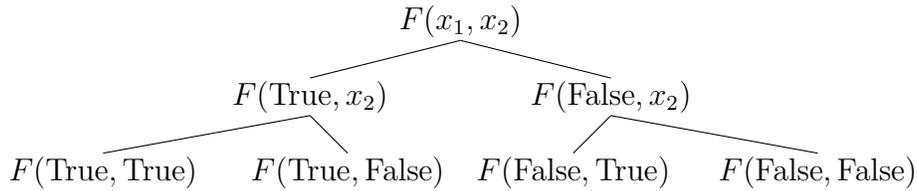

      \ctikzfig{fig-self-red}
      \caption{\label{f:self-red}The self-reducibility tree (completely unpruned) of a two-variable formula, represented generically.}
    \end{figure}

    Self-reducibility tells us that, for each node $N$ in such a
    self-reducibility tree (except the leaves, since they have no
    children), $N$ is satisfiable if and only if at least one of its
    two children is satisfiable.  Inductively, the formula at the
    root of the tree is satisfiable if and only if each level of the
    self-reducibility tree has at least one satisfiable node.  And,
    also, the formula at the root of the tree is satisfiable if and
    only if every level of the self-reducibility tree has at least one
    satisfiable node.

    How helpful is this tree?  Well, we certainly don't want to solve
    $\sat$ by checking every leaf of the self-reducibility tree.  On
    formulas with $k$ variables, that would take time at least
    $2^k$---basically a brute-force exponential-time algorithm.  Yuck!
  That isn't surprising though.  After all, the tree is really just
  listing all assignments to the formula.

  But the magic here, which we will exploit, is that the
  ``self-reducibility'' relationship between nodes and their children
  as to satisfiability will, at least with certain extra assumptions
  such as about P-selectivity, allow us to \emph{not} explore the
  whole tree.  Rather, we'll be able to prune away, quickly, all but a
  polynomially large subtree.  In fact, though on its surface this
  chapter is about four questions from complexity theory, it really is
  about tree-pruning---a topic more commonly associated with algorithms
  than with complexity.  To us, though, that is not a problem but an
  advantage.  As we mentioned earlier, complexity is largely about
  building algorithms, and
  that helps make
  complexity far more inviting and intuitive than most people realize.

  That being said, let us move on to giving a proof of
  the first challenge problem.  Namely, in this section
  we sketch a proof of the result:
  \begin{quote}
      If $\sat$ is $\p$-selective, then $\sat \in \p$.
    \end{quote}

    So assume that $\sat$ is P-selective, via (in the sense of
    Definition~\ref{d:p-sel}) polynomial-time computable function $f$.  Let us
    give a polynomial-time algorithm for $\sat$.  Suppose the input to
    our algorithm is the formula $F(x_1,x_2,\dots,x_k)$.  (If the
    input is not a syntactically legal formula we immediately reject,
    and if the input is a formula that has zero variables, e.g.,
    $\true \land \true \land \false$, we simply evaluate it and accept
    if and only if it evaluates to $\true$.)  Let us focus on $F$
    and $F$'s two children in the self-reducibility tree, as shown in Figure~\ref{f:psel}.
    \begin{figure}[tbp]
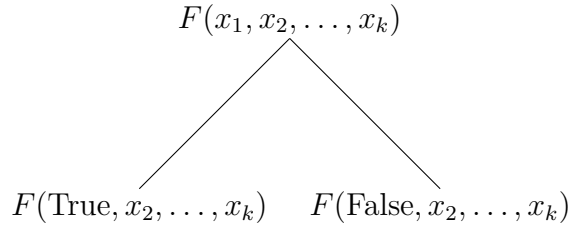

      \ctikzfig{fig-psel}
      \caption{\label{f:psel}$F$ and $F$'s two children.}
    \end{figure}
    
    Now, run $f$ on $F$'s two children.  That is,
    compute, in polynomial time,
    $f(
    F(\true,x_2,\dots,x_k),\,
    F(\false,x_2,\dots,x_k))$.  Due to the properties of 
    P-selectivity and self-reducibility, 
    note that
    the output of that application of $f$ is a formula/node that
    has the property that the original formula is satisfiable
    if and only if that child-node is satisfiable.  
    
In particular, if     $f(
    F(\true,x_2,\dots,x_k),\,
    F(\false,x_2,\dots,x_k)) =     F(\true,x_2,\dots,x_k)$
    then we know that 
    $F(x_1,x_2,\dots,x_k)$ is satisfiable if and only if
    $F(\true,x_2,\dots,x_k)$ is satisfiable.  And if $f(
    F(\true,x_2,\dots,x_k),\,
    F(\false,x_2,\dots,x_k)) \neq     F(\true,x_2,\dots,x_k)$
    then we know that 
    $F(x_1,x_2,\dots,x_k)$ is satisfiable if and only if
    $F(\false,x_2,\dots,x_k)$ is satisfiable.

    Either way, we have in time polynomial in the input's size eliminated the
    need to pay attention to one of the two child nodes, and now may
    focus just on the other one.

    Repeat the above process on the child that, as per the above, was
    selected by the selector function.  Now, ``split'' that formula by
    assigning $x_2$ both possible ways.  That will create two
    children, and then analogously to what was done above, use the
    selector function to decide which of those two children is the
    more promising branch to follow.

    Repeat this until we have assigned all variables.  We now have a
    fully assigned formula, but due to how we got to it, we know that
    it evaluates to $\true$ if and only if the original formula is
    satisfiable.  So if that fully assigned formula evaluates to
    $\true$,then we state that the original formula is satisfiable (and
    indeed, our path down the self-reducibility tree has outright
    put into our hands a satisfying assignment).  And, more
    interestingly, if the fully assigned formula evaluates to
    $\false$, then we state that the original formula is not satisfiable.
    We are correct in stating that,
    because at each iterative stage we know that if the formula we
    start that stage focused on is satisfiable, then the child the
    selector function chooses for us will also be satisfiable.

    The process above is an at most polynomial number of at most
    polynomial-time ``descend one level having made a linkage''
    stages, and so overall itself runs in polynomial time.  Thus we
    have given a polynomial-time algorithm for $\sat$, under the
    hypothesis that $\sat$ is P-selective.  This completes the proof
    sketch.

    Our algorithm was mostly focused on tree pruning.  Though
    $F$ induces a giant binary tree as to doing variable assignments
    one variable at a time in all possible ways, thanks to the
    guidance of the selector function, we walked just a
    single path through that tree.

    Keeping this flavor of approach in mind might be helpful on
    Challenge Problem~2, although that is a different problem and
    so perhaps you'll have to bring some new twist, or greater
    flexibility, to what you do to tackle that.

    And now, please pop right on back to the main body of the
    chapter, to read and tackle
    Challenge Problem~2!

  \section{Solution to Challenge Problem 2}\label{s:solution-2}    
  In this section we sketch a proof of the result:
  \begin{quote}
If there exists a tally set $T$
  such that $\sat \manyone T$, then $\sat \in \p$. 
    \end{quote}

    So assume that there exists a tally set $T$ such that
    $\sat \manyone T$.
Let $g$ be the polynomial-time computable
    function performing that reduction, in the sense of
    Definition~\ref{d:reduce}.
    (Keep in mind that we may \emph{not}
    assume that $T \in \p$.  We have no argument line in hand
    that would tell us
    that that
    happens to be true.)
    Let us give a polynomial-time algorithm for $\sat$.

    Suppose the input to
    our algorithm is the formula $F(x_1,x_2,\dots,x_k)$.  (If the
    input is not a syntactically legal formula we immediately reject,
    and if the input is a formula that has zero variables
    we simply evaluate it and accept
    if and only if it evaluates to $\true$.)

    Let us focus first on $F$.  Compute, in polynomial time,
    $g(F(x_1,x_2,\dots,x_k))$.  
    If $g(F(x_1,x_2,\dots,x_k)) \not\in \{\epsilon,0,00,\dots\}$, then
    clearly $F(x_1,x_2,\dots,x_k) \not\in\sat$, since we know that
    (a)~$T \subseteq \{\epsilon,0,00,\dots\}$ and 
    (b) $F(x_1,x_2,\dots,x_k) \in \sat \iff
    g(F(x_1,x_2,\dots,x_k)) \in T$.  So in that case, we output that 
    $F(x_1,x_2,\dots,x_k) \not\in\sat$.  Otherwise, we descend to the
    next level of the ``self-reducibility tree'' as follows.

    We consider the nodes (i.e., in this case, formulas)
    $F(\true,x_2,\dots,x_k)$
    and 
    $F(\false,x_2,\dots,x_k)$.
    Compute 
    $g(F(\true,x_2,\dots,x_k))$ and 
    $g(F(\false,x_2,\dots,x_k))$.
    If either of our two nodes in question does not, under the action just
    computed of $g$, map to a string in $\{\epsilon,0,00,\dots\}$, then
    that node certainly is not a satisfiable formula, and we can
    henceforward mentally ignore it and the entire tree (created
    by assigning more of its variables) rooted at it.  This is one
    key type of pruning that we will use: eliminating from consideration
    nodes that map to ``nontally'' strings.

    But there is a second type of pruning that we will use: If it happens to
    be the case that $g(F(\true,x_2,\dots,x_k)) \in
    \{\epsilon,0,00,\dots\}$ and 
    $g(F(\true,x_2,\dots,x_k)) = 
    g(F(\false,x_2,\dots,x_k))$, then at this point it may not be clear to 
    us whether
    $F(\true,x_2,\dots,x_k)$ is or is not satisfiable.
    However, what is clear is that 
    \[F(\true,x_2,\dots,x_k) \in \sat \iff 
    F(\false,x_2,\dots,x_k) \in \sat.\]
    How do we know this?  Since $g$ reduces $\sat$ to $T$, we know
    that 
    \[g(F(\true,x_2,\dots,x_k)) \in T \iff
    F(\true,x_2,\dots,x_k) \in \sat\]
    and 
    \[g(F(\false,x_2,\dots,x_k)) \in T \iff
    F(\false,x_2,\dots,x_k) \in \sat.\]
    By those observations, the fact that
    $g(F(\true,x_2,\dots,x_k)) = 
    g(F(\false,x_2,\dots,x_k))$, and the transitivity of 
``${\iff}$'', we indeed have that 
    $F(\true,x_2,\dots,x_k) \in \sat \iff
    F(\false,x_2,\dots,x_k) \in \sat$.  But since that says that either both or 
    neither of these nodes is a formula belonging to $\sat$, there is
    no need at all for us to further explore more than one of them, since
    they stand or fall together as to membership in $\sat$.  So
    if we have 
    $g(F(\true,x_2,\dots,x_k)) = 
    g(F(\false,x_2,\dots,x_k))$, we can mentally dismiss
    $F(\false,x_2,\dots,x_k)$---and of course also the
    entire subtree rooted at it---from all further consideration.

    After doing the two types of pruning just mentioned, we will have
    either one or two nodes left at the level of the tree---the level
    one down from the root---that we are considering.  (If we have
    zero nodes left, we have pruned away all possible paths and can
    safely reject).  Also, if $k = 1$, then we can simply check
    whether at least one node that has not been pruned away evaluates
    to $\true$, and if so we accept and if not we reject.

    But what we have outlined can iteratively be carried out in a way
    that drives us right down through the tree, one level at a time.
    At each level, we take all nodes (i.e., formulas; we will speak
    interchangeably of the node and the formula that it is
    representing) that have not yet been eliminated from
    consideration, and for each, take the next unassigned variable and
    make two child formulas, one with that variable assigned to
    $\true$ and one with that variable assigned to $\false$.  So if at
    a given level after pruning we end up with $j$ formulas, we in
    this process start the next level with $2j$ formulas, each with
    one fewer variable.  Then for those $2j$ formulas we do the
    following: For each of them, if $g$ applied to that formula
    outputs a string that is not a member of
    $\{\epsilon,0,00,\dots\}$, then eliminate that node from all
    further consideration.  After all, the node clearly is not a
    satisfiable formula.  Also, for all nodes among the $2j$ such that
    the string $z$ that $g$ maps them to belongs to $\{\epsilon,0,00,\dots\}$
    and $z$ is  mapped to by $g$ by at least one other of the $2j$ nodes,
    for each such cluster of nodes that map to the same string $z$ (of the form
    $\{\epsilon,0,00,\dots\}$) eliminate all but one of the nodes
    from consideration.  After all, by the argument given above,
    either all of that cluster are satisfiable or none of them are,
    so we can eliminate all but one from consideration, since eliminating
    all the others still leaves one that is
    satisfiable,
    if in fact the nodes in the cluster are satisfiable.

    Continue this process until (it internally terminates with a decision, or)
    we reach a level where all variables
    are assigned.  If there were $j$ nodes at the level above that after
    pruning, then
    at this no-variables-left-to-assign level we have at most $2j$ formulas.
    The construction is such that $F(x_1,x_2,\dots,x_k) \in \sat$
    if and only if at least one of these at most $2j$ variable-free formulas
    belongs to $\sat$, i.e., evaluates to $\true$.  But we can easily
    check that in time polynomial in $2j\times |F(x_1,x_2,\dots,x_k)|$.

    Is the proof done?  Not yet.  If $j$ can be huge, we're dead, as
    we might have just sketched 
    an exponential-time algorithm.  But fortunately, and this was
    the key insight in Piotr Berman's paper that proved this result, as we
    go down the tree, level by level, the tree \emph{never} grows too wide.
    In particular, it is at most polynomially wide!

    How can we know this?
    The insight that Berman (and with luck, also you!)~had is that there are not many
    ``tally'' strings that can be reached by the reduction $g$ on the
    inputs that it will be run on in our construction on a given input.
    And that fact
    will ensure us that after we do our two kinds of pruning, we have at most
    polynomially many strings left at the now-pruned level.

    Let us be more concrete about this, since it is not just the heart of
    this problem's solution, but also might well (\emph{hint!, hint!}) be
    useful when tacking the third challenge problem.

    In particular, we know that $g$ is polynomial-time computable.  So
    there certainly is some natural number $k$ such that, for each
    natural number $n$, the function $g$ runs in time at most $n^k +k$ on all
    inputs of length $n$.  Let $m = |F(x_1,x_2,\dots,x_k)|$.  Note
    that, at least if the encoding scheme is reasonable and
    we if needed do reasonable, obvious simplifications 
    (e.g., $\true \land y \equiv y$,
    $\true \lor y \equiv \true$, $\neg \true \equiv \false$, and $\neg
    \false \equiv \true$), then each formula in the tree is of
    length less than or equal to $m$.  Crucially, $g$ applied to
    strings of length less than or equal to $m$ can never output
    any string of length greater than $m^k+k$.  And so there are
    at most $m^k +k +1$ strings (the ``+\,1'' is because the empty
    string is one of the strings that can be reached) in 
    $\{\epsilon,0,00,\dots\}$ that can be mapped to by any of
    the nodes that are part of our proof's self-reducibility
    tree when the input is $F(x_1,x_2,\dots,x_k)$.
  So at each level of
    our tree-pruning, we eliminate all nodes that map to strings
    that do not belong to 
    $\{\epsilon,0,00,\dots\}$, and since we leave at most one node
    mapping to each string that is mapped to in 
    $\{\epsilon,0,00,\dots\}$, and as we just argued that there are
    at most $m^k+k+1$ of those, at the end of pruning a given level,
    at most $m^k+k+1$ nodes are still under consideration.  But
    $m$ is the length of our problem's input, so each level, after pruning,
    finishes with at most 
    at most $m^k+k+1$ nodes, and so the level after it, after
    we split each of the current level's nodes, will begin with at most 
    $2(m^k+k+1)$ nodes.  And after pruning \emph{that} level,
    it too ends up with at most 
    $m^k+k+1$ nodes still in play.  The tree indeed remains
    at most polynomially wide.

    Thus when we reach the ``no variables left unassigned'' level, we
    come into it with a polynomial-sized set of possible satisfying
    assignments (namely, a set of at most $m^k+k+1$ assignments), and 
    we know that the original formula is satisfiable if and only if at least one of
    these assignments satisfies $F$.

    Thus the entire algorithm is a polynomial number of  rounds
    (one per variable eliminated), each taking polynomial time.  So
    overall it is a polynomial-time algorithm that it is correctly
    deciding $\sat$.  This completes the proof sketch.

    And now, please pop right on back to the main body of the chapter,
    to read and tackle Challenge Problem~3!  While doing so, please
    keep this proof in mind, since doing so will be useful on
    Challenge Problem~3\ldots~though you also
    will need to discover a quite cool additional insight---the same
    one Steve Fortune discovered when he originally proved the theorem
    that is our Challenge Problem~3.

  \section{Solution to Challenge Problem 3}\label{s:solution-3}
  In this section we sketch a proof of the result:
  \begin{quote}
  If there exists a sparse set $S$
  such that $\overline{\sat} \manyone S$, then $\sat \in \p$.
    \end{quote}

    So assume that there exists a sparse set $S$ such that
    $\overline{\sat} \manyone S$.
    Let $g$ be the polynomial-time computable
    function performing that reduction, in the sense of
    Definition~\ref{d:reduce}.
    (Keep in mind that we may \emph{not}
    assume that $S \in \p$.  We have no argument line in hand
    that would tell us
    that that
    happens to be true.)
    Let us give a polynomial-time algorithm for $\sat$.

    Suppose the input to
    our algorithm is the formula $F(x_1,x_2,\dots,x_k)$.  (If the
    input is not a syntactically legal formula we immediately reject,
    and if the input is a formula that has zero variables
    we simply evaluate it and accept
    if and only if it evaluates to $\true$.)

    What we are going to do here is that we are going to mimic the proof
    that solved Challenge Problem~2.  We are going to go level by level
    down the
    self-reducibility tree, pruning at each level,
    and arguing that the tree never gets too wide---at least if we are
    careful and employ a rather jolting insight that
    Steve Fortune (and with luck, also you!)~had.

    Note that of the two types of pruning we used in the Challenge Problem~2
    proof, one applies perfectly well here.  If two or more nodes on a given
    level of the tree map under $g$ to the same string, we can eliminate
    from consideration all but one of them, since either all of them or
    none of them are satisfiable.

    However, the other type of pruning---eliminating all nodes not
    mapping to a string in
    $\{\epsilon,0,00,\dots\}$---completely disappears here.
    Sparse sets don't have too many strings per level, but the strings
    are not trapped to always being of a specific, well-known form.

    Is the one type of pruning that is left to us enough to keep the
    tree from growing exponentially bushy as we go down it?
    At first glance, it seems that exponential width
    growth is very much possible, e.g.,
    imagine the case that every node of the tree maps to a different
    string than all the others at the node's same level. Then with each
    level our tree would be doubling in size, and by its base, if we
    started with $k$ variables, we'd have $2^k$ nodes at the base
    level---clearly an exponentially bushy tree.

    But Fortune stepped back and realized something lovely.  He
    realized that if the tree ever became too bushy, \emph{then
    that itself would
    be an implicit proof that $F$ is satisfiable}!  Wow;
  mind-blowing!

  In particular, Fortune
  used the following
  beautiful reasoning.

  We know $g$ runs in
    polynomial time.  So let the polynomial $r(n)$ bound $g$'s running
    time on inputs of length $n$, and without loss of generality,
    assume that $r$ is nondecreasing.  We know that $S$ is sparse, so
    let the polynomial $q(n)$ bound the number of strings in $S$ up to
    and including length $n$, and without loss of generality, assume
    that $q$ is nondecreasing.
    
    Let  $m = |F(x_1,x_2,\dots,x_k)|$, and as before, note that all
    the nodes in our proof are of length less than or equal to $m$.  

    How many distinct strings in $S$ can be reached by applying
    $g$ to strings of length at most $m$?  On inputs of length at most
    $m$, clearly $g$ maps to strings of length at most $r(m)$.  But
    note that the number of strings in $S$ of length at most $r(m)$ is
    at most $q(r(m))$.

    Now, there are two cases.  Suppose that at each level of our tree
    we have, after pruning, at most $q(r(m))$ nodes left active.  Since
    $q(r(m))$ itself is a polynomial in the input size, $m$, that
    means our tree remains at most polynomially bushy (since levels of
    our tree are 
    never, even right after splitting a level's nodes to create
    the next level,
    wider than $2q(r(m))$).  
    Analogously to the argument of Challenge Problem~2's proof, when
    we reach the ``all variables assigned'' level, we enter it with a
    set of at most $2q(r(m))$ no-variables-left formulas such that $F$
    is satisfiable if and only if at least one of those formulas
    evaluates to $\true$.  So in that case, we easily do compute
in polynomial time whether the given input is satisfiable, analogously to the previous proof.

On the other hand, suppose that on some level, after pruning, we
    have at least $1+ q(r(m))$ nodes.  This means that at that level, we had
    at least $1+ q(r(m))$
    distinct labels.  But there are only $q(r(m))$ distinct strings
    that $g$ can possibly reach, on our inputs, that belong
    to $S$.  So at least one of the
    $1+ q(r(m))$ formulas in our surviving nodes maps to a string that
    does not belong to $S$.  But $g$ was a reduction from $\overline{\sat}$
    to $S$, so that node that mapped to a string that does not belong to
    $S$ must itself be a satisfiable formula.  Ka-\emph{zam!}  That node is
    satisfiable, and yet that node is simply $F$ with some of its
    variables fixed.  And so $F$ itself certainly is satisfiable.  We are
    done, and so the moment our algorithm finds a level that has
    $1+ q(r(m))$ distinct labels, our algorithm halts and declares that
    $F(x_1,x_2,\dots,x_k)$ is satisfiable.
    
    Note how subtle the action here is.  The algorithm is correct in
    reasoning that, when we have at least $1+q(r(m))$ distinct labels at a
    level, at least one of the still-live nodes at that level must be
    satisfiable, and thus $F(x_1,x_2,\dots,x_k)$ is satisfiable.
    However, the algorithm doesn't know a particular one of those 
at-least-$1+q(r(m))$-nodes that it
    can point to as being satisfiable.  It merely knows that at least
    one of them is.  And that is enough to allow the algorithm to act
    correctly.  (One can, if one wants, extend the above approach to
    actually drive onward to the base of the tree; what one does is
    that at each level, the moment one gets to $1+q(r(m))$ distinct
    labels, one stops handling that level, and goes immediately on to
    the next level, splitting each of those $1+q(r(m))$ nodes into two
    at the next level.  This works since we know that at least one of
    the nodes is satisfiable, and so we have ensured that at least
    node at the next level will be satisfiable.)  This completes the proof
    sketch.

    And now, please pop right on back to the main body of the chapter,
    to read and tackle Challenge Problem~4!  There, you'll be working within a
    related but changed and rather challenging setting: you'll be
    working in the realms of functions and counting.  Buckle up!

  \section{Solution to Challenge Problem 4}\label{s:solution-4}

  \subsection{Why One Natural Approach Is Hopeless}

  One natural approach would be to run the hypothetical 2-enumerator
  $h$ on the input formula $F$ and both of $F$'s $x_1$-assigned
  subformulas, and to argue that \emph{purely based on the two options
    that $h$ gives for each of those three, i.e., viewing the formulas
    for a moment as black boxes} (note: without loss of generality, we may assume
  that each of the three applications of
  the 2-enumerator has two distinct outputs; the other cases are
  even easier), we can either output $\|F\|$ or can identify at least
  one of the subformulas such that we can show a particular 1-to-1
  linkage between which of the two predicted numbers of solutions it
  has and which of the two predicted numbers of solutions $F$ has.
  And then we would iteratively walk down the tree, doing that.

  But the following
example, based on one suggested by 
  Gerhard Woeginger,
  shows
  that
  that
  is
    impossible.  Suppose $h$
    predicts
        outputs $\{0,1\}$ for $F$, and $h$ predicts outputs $\{0,1\}$ for the left subformula,
    and $h$ predicts outputs 
        $\{0,1\}$
        for the right subformula.  That is, for each, it says ``this formula
        either has zero satisfying assignments or has exactly one
        satisfying assignment.''
In this case, note that the values of the root
    can't be, based solely on the numbers the enumerator output, linked 1-to-1 to those of the left
    subformula, since $0$ solutions for the left subformula
    can correspond to a root value of~0~($0+0=0$) or to a root value of~1~($0+1=1$).  The same clearly also holds for the right subformula.

    The three separate number-pairs just don't have enough information
    to make the desired link!  But don't despair: we can
    make~$h$ help us far more powerfully than was done above!

  \subsection{XYZ Idea/Statement}\label{ss:XYZ}

  To get around the obstacle just mentioned, we can try to trick the
  enumerator into giving us \emph{linked/coordinated} guesses!  Let us
   see how to do that.

  What I was thinking of, when I mentioned XYZ in the food-for-thought
  hint (Section~\ref{ss:food}), is the fact that we can efficiently combine two Boolean
  formulas into a new one such that from the number of satisfying
  assignments of the new formula we can easily ``read off'' the number
  of satisfying assignments of both the original formulas.  In 
 fact, it turns out that we can do the combining in such a way that if we
  concatenate the (appropriately padded as needed) bitstrings
  capturing the numbers of solutions of the two formulas, we get the
  (appropriately padded as needed) bitstring capturing the number of
  solutions of the new ``combined'' formula.  We will, when
  $F$ is a Boolean formula, use $\|F\|$
  to denote the number of satisfying assignments of $F$.

\begin{lemma}\label{l:combining}
There are polynomial-time computable functions $\mathrm{combiner}$ and
$\mathrm{decoder}$ such that for any Boolean formulas $F$ 
and $G$,
$\mathrm{combiner}(F,\, G)$ is a Boolean formula and $\mathrm{decoder}(F,G,
\|\mathrm{combiner}(F,G)\|)$ prints $\|F\|,\|G\|$.
\end{lemma}

\begin{proof}[Proof Sketch]
Let $F=F(x_1,\ldots,\,x_n)$ and
$G=G(y_1,\ldots,\,y_m)$, where $x_1,\ldots,\,x_n,y_1,\ldots,\,y_m$
are all distinct.  Let $z$ and $z^\prime$ be two
new Boolean variables.  Then
\[
 H=(F\wedge z)\vee (\bar{z}\wedge x_1 \wedge \cdots \wedge
 x_n \wedge G \wedge z^\prime )
\]
gives the desired combination, since
$\|h\|=\|f\|2^{m+1}+\|g\|$ and
$\|g\|\leq 2^m$.
\end{proof}

We can easily extend this technique to combine three, four, or even
polynomially many formulas.

  \subsection{Invitation to a Second Bite at the Apple}
Now that you have in hand the extra tool that is
  Lemma~\ref{l:combining}, this would be a great time, unless
  you already found a solution to the fourth challenge problem,
  to try again to solve the problem.
  My guess is that if you did not already solve the fourth challenge
  problem, then the ideas
  you had while trying to solve it will stand you in good stead when you
  with the combining lemma in hand 
  revisit the problem.

  My suggestion to you would be to work again on proving Challenge
  Problem~4
    until either you find a proof, or you've put in at least 15 more minutes
    of thought, are stuck, and don't think that more time will be helpful.

    When you've reached one or the other of those states, please go on to
    Section~\ref{ss:proof-4} to read a proof of the theorem.

  \subsection{Proof Sketch of the Theorem}\label{ss:proof-4}

Recall that we are trying to prove:
  \begin{quote}
  If $\sharpsat$ is has a polynomial-time $2$-enumerator, then there is
  a polynomial-time algorithm for $\sharpsat$.
  \end{quote}

Here is a quick proof sketch.
  Start with our input formula, $F$, whose number of solutions
  we wish to compute in polynomial time.  
  If~$F$ has
  no variables, we can simply directly output the right number
  of solutions, namely,~1 (if~$F$ evaluates to $\true$), or~0 (otherwise).
  Otherwise, self-reduce formula $F$ on its first variable.
  Using the XYZ trick, twice,
  combine the original formula and the two subformulas into a
  single formula, $H$,
  whose number of solutions gives the number of
solutions of all three.
For example, if our three formulas are $F
= F(x_1
x_2, x_3, \dots)$, $F_\mathit{left}
= F(\true, x_2, x_3, \dots)$,
  and $F_\mathit{right}= F(\false, x_2, x_3, \dots)$, our combined formula can be 
  \[H = \mathrm{combiner}(F,\mathrm{combiner}(F_\mathit{left},F_\mathit{right})),\]
  and the decoding process
 is clear from this and
 Lemma~\ref{l:combining}
(and its proof).
Run the 2-enumerator on $H$.
If either of $H$'s output's two decoded guesses are inconsistent
($a \neq b + c$), then ignore that line and the other one is the truth.  If
both are consistent and agree
  on $\|F\|$, then we're also done.  Otherwise, the two guesses must each be
  internally consistent and 
  the two guesses must disagree on $\|F\|$, and so it follows that the two 
guesses differ in their
claims about at least one of $\|F_\mathit{left}\|$ and
$\|F_\mathit{right}\|$.
Thus if we know
  the number of solutions of that one, shorter formula, we know the number
  of solutions of $\|F\|$.

  Repeat the above on \emph{that} formula, and so on, right on
  down the three, and then (unless
  the process resolves internally or ripples back up earlier)
  at the end we have reached a zero-variable formula and for it
  we by inspection will know how many solutions it has (either 1 or 0), and so using
  that we can ripple our way 
  all the way back up through the tree, using our linkages between
  each level and the next, and thus we now have computed $\|F\|$.
  The entire process is a polynomial
  number of polynomial-time actions, and so runs in polynomial time
  overall.

  That ends the proof sketch, but let us give an 
  example regarding the key step from the proof sketch, as that will
  help make 
  clear what is going on.

\newlength{\myheight}
\newlength{\mywidth}
  \settowidth{\mywidth}{the Guesses}
\addtolength{\mywidth}{18pt}
  \settoheight{\myheight}{Which of the Guesses}
\begin{center}
\begin{tabular}{c  c c c}
\parbox[c][\myheight]{\mywidth}{\centering Which of\\the Guesses} &  $\|F(x_1, x_2, x_3, \dots)\|$  & $\|F(\true, x_2, x_3, \dots)\|$
  &
    $\|F(\false, x_2, x_3, \dots)\|$\\[8pt]
  \toprule
  First &  100 & 83 & 17 \\
  Second & 101 & 85 & 16
\end{tabular}
\end{center}

In this example, note that we can conclude that
$\|F\| = 100$ if $\|F(\false, x_2, x_3, \dots)\| = 17$, and 
$\|F\| = 101$ if $
\|F(\false, x_2, x_3, \dots)\|  = 16$;
and we know that $ \|F(\false, x_2, x_3, \dots)\|
 \in \{16,17\}$.

 So we have in polynomial time completely linked
$ \|F(x_1, x_2, x_3, \dots)\|$
to
the issue of the number of satisfying assignments of the
(after simplifying) shorter 
formula
$F(\false, x_2, x_3, \dots)$.  This completes our example
of the key linking step.
\bibliographystyle{alpha}
\newcommand{\etalchar}[1]{$^{#1}$}

\end{document}